\RequirePackage{fix-cm}
\documentclass[smallextended]{svjour3}       % onecolumn (second format)
\smartqed  % flush right qed marks, e.g. at end of proof
\usepackage{graphicx}
\usepackage{amssymb,amsmath}

\usepackage{fontawesome5}
\usepackage{graphicx}
\usepackage{color}
\usepackage{colortbl}
\usepackage{bm}
\usepackage{complexity}
\usepackage{framed}
\usepackage{enumitem}
\usepackage[caption=false]{subfig}
\usepackage{float}
\usepackage{hyperref}
\usepackage{cite}
\usepackage{array}
\usepackage{wasysym}
\usepackage[colorinlistoftodos,prependcaption,textsize=tiny]{todonotes}
\usepackage[latin1]{inputenc}
\usepackage{soul}

\usepackage{xcolor}
\usepackage{listings}
\usepackage{forloop}
\usepackage{calc}
\usepackage{ifthen}
\usepackage{booktabs}
\usepackage{todonotes}
\usepackage{amsmath}
\usepackage{complexity}
\usepackage{multicol}

\usepackage{inconsolata}
 \hypersetup{
     colorlinks, linkcolor={dark-blue},
    citecolor={mid-green}, urlcolor={dark-blue}}

  \definecolor{mid-green}{rgb}{0.15,0.65,0.15}
 \definecolor{dark-green}{rgb}{0.15,0.25,0.15}
 \definecolor{dark-red}{rgb}{0.7,0.15,0.15}
 \definecolor{dark-blue}{rgb}{0.15,0.15,0.9}
 \definecolor{medium-blue}{rgb}{0,0,0.5}
 \definecolor{gray}{rgb}{0.5,0.5,0.5}
 \definecolor{color-Ig}{rgb}{0.15,0.7,0.15}
 \definecolor{darkmagenta}{rgb}{0.30, 0.0, 0.30}
 \definecolor{blue}{rgb}{0.15,0.15,0.9}

\newcommand{\bjpmp}{{\sc BJIP}\xspace}

\newtheorem{innercustomgeneric}{\customgenericname}
\providecommand{\customgenericname}{}
\newcommand{\newcustomtheorem}[2]{%
  \newenvironment{#1}[1]
  {%
   \renewcommand\customgenericname{#2}%
   \renewcommand\theinnercustomgeneric{##1}%
   \innercustomgeneric
  }
  {\endinnercustomgeneric}
}

\newcustomtheorem{customprop}{Proposition}
\newcustomtheorem{customlmm}{Lemma}
\newcustomtheorem{customdef}{Definition}
\newcustomtheorem{customthm}{Theorem}

\usepackage{tikz}
 \usetikzlibrary{calc}
 \usepackage{boxedminipage}
\usepackage{framed}
\usepackage{xspace}
\usepackage{xifthen}
 \usepackage{tabularx}

\newlength{\RoundedBoxWidth}
\newsavebox{\GrayRoundedBox}
\newenvironment{GrayBox}[1]%
   {\setlength{\RoundedBoxWidth}{.93\textwidth}
    \def\boxheading{#1}
    \begin{lrbox}{\GrayRoundedBox}
       \begin{minipage}{\RoundedBoxWidth}}%
   {   \end{minipage}
    \end{lrbox}
    \begin{center}
    \begin{tikzpicture}%
       \node(Text)[draw=black!20,fill=white,rounded corners,%
             inner sep=2ex,text width=\RoundedBoxWidth]%
             {\usebox{\GrayRoundedBox}};
        \coordinate(x) at (current bounding box.north west);
        \node [draw=white,rectangle,inner sep=3pt,anchor=north west,fill=white]
        at ($(x)+(6pt,.75em)$) {\boxheading};
    \end{tikzpicture}
    \end{center}}

\newenvironment{defproblemx}[2][]{\noindent\ignorespaces%
                                \FrameSep=6pt%
                                \parindent=0pt%
                \vspace*{-1.5em}
                \ifthenelse{\isempty{#1}}{%
                  \begin{GrayBox}{\textsc{#2}}%
                }{%
                  \begin{GrayBox}{\textsc{#2} parameterized by~{#1}}%
                }
                \newcommand\Input{Input:}%
                \begin{tabular*}{\textwidth}{@{\hspace{.1em}} >{\itshape} p{1.8cm} p{0.8\textwidth} @{}}%
            }{
                \end{tabular*}%
                \end{GrayBox}%
                \ignorespacesafterend
            }

\newcommand{\defproblemaOPT}[3]{%
  \begin{defproblemx}{#1}
    {\bf Instance:}  & #2 \\
    {\bf Goal:} & #3
  \end{defproblemx}
}%

\definecolor{azul1}{RGB}{183,200,196}

 \hypersetup{
     colorlinks, linkcolor={dark-blue},
    citecolor={mid-green}, urlcolor={dark-blue}}

  \definecolor{mid-green}{rgb}{0.15,0.65,0.15}
 \definecolor{dark-green}{rgb}{0.15,0.25,0.15}
 \definecolor{dark-red}{rgb}{0.7,0.15,0.15}
 \definecolor{dark-blue}{rgb}{0.15,0.15,0.9}
 \definecolor{medium-blue}{rgb}{0,0,0.5}
 \definecolor{gray}{rgb}{0.5,0.5,0.5}
 \definecolor{color-Ig}{rgb}{0.15,0.7,0.15}
 \definecolor{darkmagenta}{rgb}{0.30, 0.0, 0.30}
 \definecolor{blue}{rgb}{0.15,0.15,0.9}

\newcommand{\ub}[2]{{\underbrace{#1}_{#2}}}

\usepackage[english,algoruled,vlined,longend,ruled,linesnumbered,noend]{algorithm2e}

\def\qed{{\hfill\hbox{\rlap{$\sqcap$}$\sqcup$}}}

\usepackage{url}

\begin{document}

\title{Binary Jumbled Indexing: Suffix tree histogram\thanks{
Financial support from FAPERJ E-26/201.372/2022, and CNPq 406173/2021-4. 
An extended abstract of this work 
% , with incomplete proofs, 
was recently presented in COCOON 2024~\cite{COCOON2024}.}
}
% \subtitle{Binary Jumbled Pattern Matching: Suffix tree indexing}


\author{Lu\'is Cunha         \and
        M\'ario Medina* %etc.
}


\institute{L. Cunha \at
              Universidade Federal Fluminense \\
              % Tel.: +123-45-678910\\
              % Fax: +123-45-678910\\
              \email{lfignacio@ic.uff.br}           %  \\
           \and
           M. Medina \at
              Universidade Federal Fluminense
              \email{mmedina@id.uff.br,mazen.mario@gmail.com}             \\~\\
              * Corresponding author
}

\date{Received: date / Accepted: date}

\maketitle

\begin{abstract}
Given a binary string $\omega$ over the alphabet $\{0, 1\}$, a vector $(a, b)$ is a Parikh vector if and only if a factor of $\omega$ contains exactly $a$ occurrences of $0$ and $b$ occurrences of $1$. Answering whether a vector is a Parikh vector of $\omega$ is known as the Binary Jumbled Indexing Problem (\bjpmp) 
% or the Histogram Indexing Problem~\cite{chan2015clustered}. 
or the Histogram Indexing Problem. 
Most solutions to this problem rely on an $O(n)$ word-space index to answer queries in constant time, encoding the Parikh set of $\omega$, i.e., all its Parikh vectors. Cunha et al. (\emph{Combinatorial Pattern Matching}, 2017) introduced an algorithm (\emph{JBM2017}), which computes the index table in $O(n+\rho^2)$ time, where $\rho$ is the number of runs of identical digits in $\omega$, leading to $O(n^2)$ 
% in the average case. 
in the worst case. 
We prove that the average number of runs $\rho$ is $n/4$, confirming the quadratic behavior
also in the average-case. 
We propose a new algorithm, \emph{SFTree}, which uses a suffix tree to remove duplicate substrings. Although \emph{SFTree} also has an average-case complexity of $\Theta(n^2)$ due to the fundamental reliance on run boundaries, it achieves practical improvements by minimizing memory access overhead through vectorization. The suffix tree further allows distinct substrings to be processed efficiently, reducing the effective cost of memory access. As a result, while both algorithms exhibit similar theoretical growth, \emph{SFTree} significantly outperforms others in practice. Our analysis highlights both the theoretical and practical benefits of the \emph{SFTree} approach, with potential extensions to other applications of suffix trees.

\end{abstract}

\section{Introduction} \label{intro}

    The \emph{binary jumbled indexing} problem is presented as follows: We are given a binary string $\omega$ over the alphabet $\{0, 1\}$ and asked to determine whether there exists a substring of size $r$ and $b$ $1$s. Such substring could be represented by a \emph{Parikh vector} of $\omega$, something that appears frequently in computational biology~\cite{benson2003composition,eres2004permutation}, as do jumbled patterns in the context of graphs and other structures~\cite{cicalese2013indexes,lacroix2006motif,song2021fast}. This problem has aroused much interest, as seen in a few approaches~\cite{cicalese2009searching,burcsi2012approximate,badkobeh2013binary,cunha2017faster}, and since the queries and their quantity can be arbitrary, the interest is for the problem of \emph{indexing binary strings for jumbled pattern matching}, as described below:

    \defproblemaOPT{Binary Jumbled Indexing problem (\bjpmp)}{A finite binary string $\omega$ of length $n$ over the alphabet $\{0, 1\}$.}{Construct an index table to answer queries efficiently: for integers $a, b \geq 0$, does $\omega$ have a factor with $a$ $0$s and $b$~$1$s?}

    The \bjpmp, also referred to as histogram indexing, is equivalent to determining the prefix normal form (PNF) of a binary string. The PNF, in turn, corresponds to an $O(n)$ bit space encoding of the index, providing a compact representation of all Parikh vectors. Additionally, Chan and Lewenstein~\cite{chan2015clustered} demonstrated that the \bjpmp is computationally equivalent to the (min,+) convolution problem, further highlighting its relevance in combinatorial pattern matching.

    Although working with a binary alphabet is a restriction of the general arbitrary alphabet size case, it offers the advantage of enabling $O(1)$ query time for the \bjpmp. For larger alphabets, this indexing format may result in increased query times, as the complexity of representing and accessing the Parikh set grows with the size of the alphabet, as shown by Chan and Lewenstein~\cite{chan2015clustered}. Nonetheless, Our proposed algorithm is not limited to binary strings and could be adapted to support other indexing formats, potentially extending its applicability to strings over larger alphabets, albeit with different trade-offs in performance and complexity.

    An index, in this paper, is a table constructed for a word of length $n$ over the binary alphabet that can determine the existence of substrings with a given number of $1$s. Thus, the \bjpmp asks us to preprocess a binary string such that later, given a number of $0$s and a number of $1$s, we can quickly report whether there exists a substring with those numbers of $0$s and $1$s and, optionally, return the position of at least one such substring. Direct preprocessing algorithms take quadratic time and other approaches reduced that time complexity to $O (n^2 / \log n)$~\cite{moosa2010indexing}, $O (n^2 / \log^2 n)$~\cite{moosa2012sub}, \(O(n^2 / 2^{\Omega (\sqrt{\log n / \log \log n})})\)~\cite{hermelin2014binary} and finally $O (n^{1.859})$ with randomization or $O (n^{1.864})$ without~\cite{chan2015clustered}. Other randomized approaches are considered with subquadratic construction time~\cite{kociumaka2017efficient}. Related problems were also described, as lower bounds on reporting all certificates of a query and pattern matching with mismatches~\cite{afshani2020lower}. Despite the existence of truly subquadratic algorithms, our approach offers the advantage of avoiding recursion and optimizing memory access, making it potentially more suitable for applications where these factors significantly influence performance, depending on the available resources and the programming environment.
    
    Cunha et al.~\cite{cunha2017faster} proposed an algorithm that runs in $O(n+ \rho^2)$ time and $O(n)$ words of space. Furthermore, they showed how we can either keep the same bounds and store information that lets the index return the position of one match, or keep the same time bound and use only $O(n)$ bits of space. The algorithm in ~\cite{cunha2017faster} matches one of the two algorithms proposed by Giaquinta and Grabowski~\cite{giaquinta2013new} with the parameter $k=1$, for which one runs in $O(\rho^2 \log k + n/k)$ time, produces an index that uses $O(n/k)$ extra space and answer queries in $O(\log k)$ time, and another one that runs in $O(n^2\log^2w/w)$ time, where $w$ is the size of a machine word.

\paragraph{Contributions.}

    \begin{itemize}

        \item 
        We provide the design, analysis and implementation of a new algorithm for constructing index tables from strings with both theoretical and practical implications.

        \item
        For the theoretical side, we prove that the average number of runs grows linearly with $n$, therefore when algorithms to build index tables are based on runs, it is not possible to develop a faster strategy than $\rho^2$.
        
        \item 
        Moreover, we show that our approach can be much faster than the one proposed in~\cite{cunha2017faster}. The time complexity of the former fluctuates between quadratic and linear time, while the latter stays in quadratic time.

        \item
        For the practical side, we compare our approach with the very simple and fast one proposed in~\cite{cunha2017faster}. By using a suffix tree to store information from strings, our algorithm presents advantages depending on the number of repeated substrings and the interest in using the suffix tree for other purposes.

        \item
        Using vectorization instead of iteration or recursion offers a substantial advantage, as memory access time is significantly more costly than processing time. Our algorithm achieves memory access in $O(n)$, which allows it to outperform subquadratic alternatives for all but exceptionally large inputs.
        
        % \item 
        % Furthermore, when dealing with multiple words it is possible to create a generalized suffix tree without increasing the time complexity for building the index table, while reducing the space complexity for separated suffix trees and serving many other applications e.g. finding the longest common subsequence (LCS).

    \end{itemize}

\paragraph{Organization.}

    \autoref{sec:pre} provides the preliminaries for the \bjpmp, detailing the connections between Parikh sets, prefix normal forms, the \bjpmp itself, and the simple yet efficient algorithm previously introduced by Cunha et al.~\cite{cunha2017faster}. Due to its shared steps with our algorithm, \autoref{cunha} explains the workings of the JBM2017 algorithm~\cite{cunha2017faster}. \autoref{suffixtree} explains how strings can be encoded and how suffix trees are utilized to construct index tables. \autoref{algorithm} outlines our proposed indexing algorithm, proves its time complexity for both the worst and average cases, and establishes that the number of runs grows linearly with respect to the input. Finally, \autoref{results} presents practical results, comparing execution times and highlighting the advantages of our approach.

%\vspace{-.2cm}
\section{Preliminaries}\label{sec:pre}
%\vspace{-.3cm}

    \paragraph{Parikh set and Parikh vectors.} \label{parikh}

        Let $\Sigma = \{ 0, 1, 2, \ldots \}$ be a finite alphabet and $\omega$ be a word over $\Sigma$, i.e., a finite sequence of characters from the alphabet. Given a vector $\pi = (\pi_0, \pi_1, \pi_2, \ldots)$, $\pi$ is said to be a Parikh vector of $\omega$ if and only if there exists a substring of $\omega$ where, for each $\sigma$ in $\Sigma$, $\pi_\sigma$ is the number of occurrences of $\sigma$ in that substring. It is easy to see that for $\pi$ to be a Parikh vector, the length of such a substring must be the sum of all elements of $\pi$, and the dimension of $\pi$ must be equal to the length of $\Sigma$. For instance, given the word $011$ over the binary alphabet $\{0, 1\}$, then $\pi = (1, 1)$ is a Parikh vector of that word, since it is possible to find a substring with one $0$ and one $1$. It is important to consider that a single vector can match multiple and different substrings; that is, a Parikh vector $(2,3)$ matches each of these substrings: $00111, 01011, 01101, 01110, 10011, 10101, 10110, 11001, 11010, 11100$.

        The Parikh set of a word $\omega$ is a set of all its Parikh vectors, denoted as $\Pi(\omega)$. For example, a binary word of the form $01101$ has these and only these Parikh vectors: 
        
        $
        \Pi = \{ (1, 1), (0, 2), (1, 2), (1, 3), (2, 2), (2, 3) \}.
        $
        
        Formally, we denote:
        $
        \pi(\varepsilon) = (|\varepsilon|_{\sigma})_{\sigma \in \Sigma},
        $ 
        where $|\varepsilon|_{\sigma}$ is the number of occurrences of $\sigma$ in $\varepsilon$, substring of $\omega$. 
        $
        \Pi(\omega) = \{ \pi(\varepsilon_i) \}
        $, 
        where $\varepsilon_i$ is each substring found in $\omega$. 
        Therefore a function $f: \varepsilon_i \rightarrow \Pi $ is surjective.

    \subsection{Binary Jumbled Indexing} \label{bjpm}

        As mentioned, the \bjpmp involves preprocessing a binary string $\omega$ over $\Sigma = \{0, 1\}$ to construct an index that allows answering whether a vector $\pi$ is a Parikh vector of $\omega$ in constant time. Since no sublinear-time algorithm is known to check a single vector directly, preprocessing and indexing offer a faster solution at the expense of increased space complexity. These methods leverage the interval property to optimize the index size while ensuring efficient query responses, as detailed below.

        The \emph{interval property} of a binary string ensures that, if $x_1$ is the least occurrences of $1$s and $x_2$ is the most occurrences of $1$s for a specific length $l$, then it is possible to find a substring with length $l$ and $o$ occurrences of $1$s if and only if $x_1 \leq o \leq x_2$~\cite{burcsi2012approximate,cicalese2009searching,cunha2017faster}.
        
        As established by Badkobeh et al.~\cite{badkobeh2013binary}, we can build an index with the least and the most occurrences of $1$s --- for each length --- to answer in $O(1)$ if a certificate is in between those values.
        
        Let $max_1(l)$ be the maximum number of $1$s in any substring of length $l$ and $T_{max_1} = [t_1, t_2, t_3 \ldots]$ an array containing the $max_1(i)$ for each $t_i$. Analogously, we have an array $T_{min_1}$ for the minimum number of $1$s. Then, given the word $11011001$:
        
        $
            \begin{array}{cccc}
            11011001  & \rightarrow & T_{max_1}  = [1, 2, 2, 3, \bold{4}, 4, 4, 5], & T_{min_1} = [0, 0, 1, 2, \bold{2}, 3, 4, 5].
            \end{array}
        $
        
        Taking the fifth position of each table we know that every substring of length $5$ has $o$ occurrences of $1$s if and only if $2 \leq o \leq 4$.        
        Furthermore, we do not need to create an index with the maximum and the minimum number of $0$s, for in the binary alphabet each non $0$ is necessarily a $1$. We can derive the maximum and minimum number of $0$s from the difference between the length, and the minimum and maximum of $1$s, resp.:
        
        $
            \begin{array}{cc}
                    T_{max_0}[i] = i - T_{min_1}[i],  &  T_{min_0}[i] = i - T_{max_1}[i].
            \end{array}
        $

    \paragraph{Prefix normal forms and Prefix normal words.} \label{pnw}

        A \emph{word prefix} is an $l$-sized substring of that word with the same $l$ first characters, in the same order. In~\cite{burcsi2017prefix}, they define prefix normal word (PNW) and prefix normal form (PNF) in binary strings, which can be a $1$-prefix normal word or a $0$-prefix normal word. 
        A \emph{$1$-prefix normal word} $\omega$ is such that, for every length $1 < l < |\omega|$, an $l$-sized substring with the maximum number of $1$s is found as a prefix of $\omega$. 
        A $0$-prefix normal word is defined analogously. See~\cite{burcsi2017prefix} for more definitions.

        They demonstrated that, for every word $\omega$, it is possible to find a $1$-prefix normal word with the same index tables, referred to as the prefix normal form (PNF) of $\omega$. They also proved that every word has a unique PNF, and the set of all words sharing the same PNF is defined as a $1$-prefix equivalence class.

        For example $\omega = 1101001$ is a PNW, but its inverse (word written in reverse order) is not. For $\omega^{-1} = 1001011$, we have $T_{max_1}[3] = 2$, e.g. 101 and 011, but neither of them is a prefix of $\omega^{-1}$. Since both have the same $T_{max_1}$, 
        % table, 
        then $\omega$ is both the PNF of $\omega^{-1}$ and its own PNF. A PNF of a word is a PNW with the same $T_{max_1}$. 
        % table. 
        These definitions are important because, with a PNW, it is possible to create an index in $O(n)$ time by reading the number of $1$s for each prefix size.

% %\vspace{-.5cm}
    \subsection{Cunha et al.'s algorithm} \label{cunha}
% %\vspace{-.2cm}

        Cunha et al.~\cite{cunha2017faster} developed an algorithm for jumbled indexing of binary strings and showed that, since binary strings have the interval property, it is possible to jump between \emph{runs of $1$s}, which are defined as each sequence of consecutive $1$s. They save a single witness for every $l$-sized substring with a maximum number of $1$s. For example, $\omega = 1100101$ would be indexed for each substring starting at the beginning of a run of $1$s and finishing at the end of a run of $1$s:

        \begin{center}
        $
        \begin{array}{llcc}
             \bold{11}00101 & 11 & \mapsto & T_{max_1}[2] = 2 \\
             \bold{11001}01 & 11001 & \mapsto & T_{max_1}[5] = 3 \\
             \bold{1100101} & 1100101 & \mapsto & T_{max_1}[7] = 4 \\
             1100\bold{1}01 & 1 & \mapsto & T_{max_1}[1] = 1 \\
             1100\bold{101} & 101 & \mapsto & T_{max_1}[3] = 2 \\
             110010\bold{1} & 1 & \mapsto & T_{max_1}[1] = 1 \\
        \end{array}
        $
        \end{center}
        
        To fill out the rest of the index table, they use adjacent values and the interval property. Hence, if $T = [t_1, t_2, t_3 ...]$ is the index table where $t_i$ is the maximum number of $1$s in any substring of length $i$, then $t_i \leq t_{i+1} \leq t_i + 1$, and we can pass over $T$ right-to-left and left-to-right, assigning the maximum value between the current value and the least possible one, completing the index table. They proved that:         
        
        {%\centering
        $
        \begin{array}{lcccllllcccl}
             T[i+1]-1 & \leq & T[i] & \leq & T[i+1] &, & T[i-1] & \leq & T[i] & \leq & T[i-1]+1. 
        \end{array}
        $
        
        }

        The process of filling out the index table by using adjacent values and assigning the maximum possible value, we define as "\emph{windowizing}" the index table. This operation is implemented as a separate function in the code presentation of our algorithm.
        
        When trying to reduce the space used for the index table, they showed that we can transform an $O(n)$ word into an $O(n)$ bit index table by assigning $1$ if $t_{i+1} > t_i$, or $0$ if $t_{i+1} = t_i$. For example, considering $\omega = 010101110101$:
        
        $
        \begin{array}{cc}
             T_{word} = [1,2,3,3,4,4,5,5,6,6,7,7], & T_{bit} = [1,1,1,0,1,0,1,0,1,0,1,0].
        \end{array}
        $
        
        Here, we have a synergy between the PNF definition and the bit-encoded index table, for the sequence in the bit-encoded table is the $1$-prefix normal form of the word indexed. It is clear that trying to solve the \bjpmp using index tables is equivalent to finding the prefix normal forms of a given word.
        
        $
        \begin{array}{ccc}
            PNF_1(\omega) = 111010101010 & \rightarrow & T_{bit} = [1,1,1,0,1,0,1,0,1,0,1,0].
        \end{array}
        $

        % \medskip
        
        % The algorithm developed in~\cite{cunha2017faster} runs in $O(n+\rho^2)$ time, 
        The time complexity of the algorithm developed in~\cite{cunha2017faster} is $O(n+\rho^2)$, 
        where $\rho$ is the number of runs, which is $O(n^2)$ when $\rho$ approaches $n$. The worst case for this algorithm is given by the string $1010101\cdots$. The algorithm runs from the start of each run of $1$s to the end of the next run; therefore, it will index each occurrence of $1$, $101$, $10101$, $1010101 \cdots$. Since we have $\rho = \frac{n}{2}$, then the time complexity can be written in terms of $n$: $O(n+{(\frac{n}{2})}^2) = O(n^2)$.

        Since repeated occurrences do not change our index table, there is no upside to indexing each pattern more than once; therefore, we can get rid of repeated occurrences. In the next section, we present an algorithm that uses a suffix tree to build string patterns without repetition to reduce the time when reading each substring enclosing runs of $1$s.

\section{Suffix tree and special pattern encode} \label{suffixtree}

    A suffix tree is a specialized data structure utilized to store a list of strings~\cite{ukkonen1995line}. By design, the suffix tree incorporates each substring pattern only once, making this attribute particularly advantageous for our needs. 

    Suffix links, commonly utilized in suffix trees to enable efficient traversal and pattern matching by connecting nodes representing suffixes of the same prefix, are typically beneficial for skipping redundant computations. However, since our algorithm explicitly processes each pattern during the table construction, the suffix link structure does not contribute to efficiency in this context and is therefore omitted.

    \paragraph{Special pattern encode.} \label{encode}

        In the context of \bjpmp, our focus lies solely on the frequency of occurrences for each character of the alphabet. Hence, each consecutive run of a character can be represented by the length of its repetition. For instance, consider the string $110111001$; its special pattern encoding reflects the count of repetitions for each digit: 
        
        $
            \begin{array}{cc}
            \ub{11}{2}\ub{0}{1}\ub{111}{3}\ub{00}{2}\ub{1}{1} \ , &  110111001 \mapsto 2\ 1\ 3\ 2\ 1.\\
            \end{array}
        $

        The implementation of the proposed \emph{SFTree} algorithm has been modified to accurately differentiate single-digit values from concatenated ones, ensuring precise parsing. Furthermore, the algorithm has been adapted to record the starting digit of each suffix, enabling the index to efficiently support queries about digit counts.

% %\vspace{-1cm}
    \paragraph{Building the suffix tree.} \label{buildingtree}

        The Ukkonen's algorithm for building suffix trees is widely known for its time complexity, which is $O(n)$~\cite{ukkonen1995line}. Its essence relies on saving the initial position of each branch created instead of the entire substring. It is important to highlight that some structures of the Ukkonen's algorithm are not necessary for us, such as the suffix links.

        For example, considering the string $1101100110111011000110111$, let us build its suffix tree. First, we count each digit repetition to construct the special pattern encoding:
        $
        \begin{array}{c}
             1101100110111011000110111 \mapsto 2\ 1\ 2\ 2\ 2\ 1\ 3\ 1\ 2\ 3\ 2\ 1\ 3 %bAbBbAcAbCbAc. 
        \end{array}
        $

        Now, we use the Ukkonen's algorithm to build its suffix tree, without using suffix links. See \autoref{fig:partial-suffix-tree} for an example.

        It is important to note that Ukkonen's algorithm inherently saves the positions of the certificates during the construction of the suffix tree. These positions can be retrieved and optionally indexed alongside the main index table. This feature allows the algorithm to not only confirm the existence of patterns but also efficiently return their locations when required.

%\vspace{-.8cm}
        \begin{figure}[!h]
            \centering
            \includegraphics[width=6.5cm]{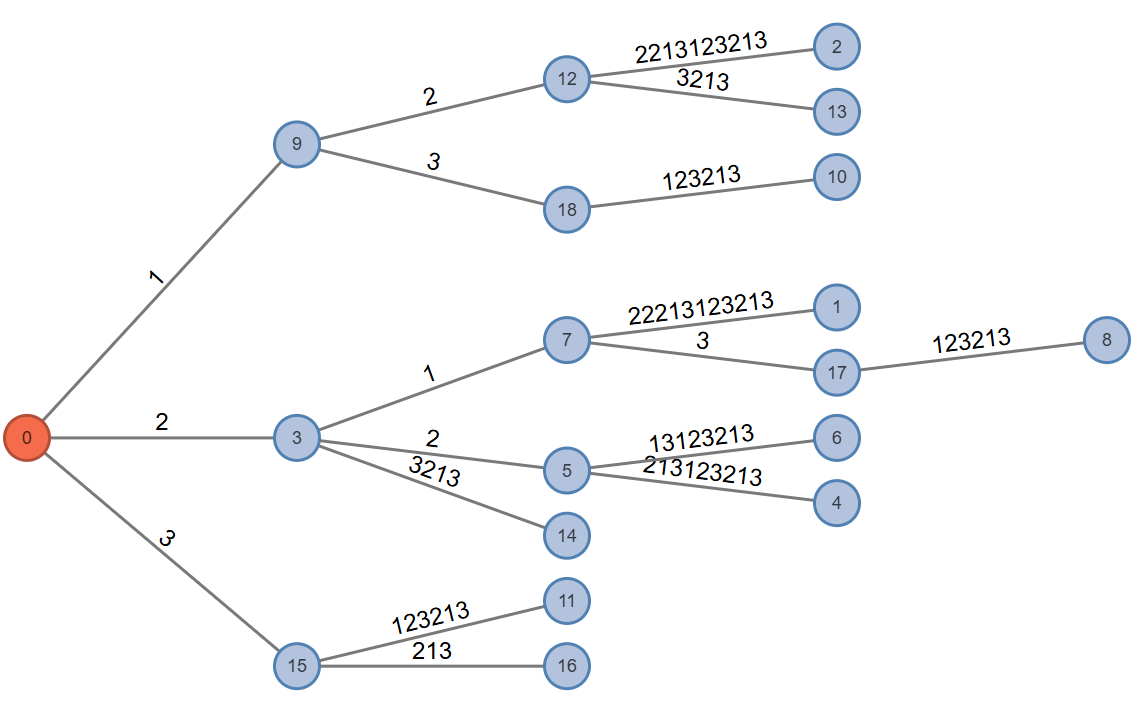}
            \caption{Build in: \url{https://brenden.github.io/ukkonen-animation/}. The node labels correspond to the steps of Ukkonen's algorithm. 
            \label{fig:partial-suffix-tree}}
        \end{figure}
%\vspace{-1.2cm}

\section{Binary Jumbled Indexing: Algorithm} \label{algorithm}

    Now, we describe our strategy for building the index of a given binary string. Essentially, we begin with a binary string. We then utilize $O(n)$ time to construct its special pattern encoding and an additional $O(n)$ time to create its suffix tree using Ukkonen's algorithm. Afterward, we extract each possible factor from the suffix tree. This involves performing a pre-order traversal of the suffix tree and indexing from the parent node up to the current node, described in \autoref{alg:suffix_tree}.

    Let us illustrate with some steps from the suffix tree in \autoref{fig:partial-suffix-tree}. Select node~$3$. Its substring is '2', which has total length equals 2. Indexing:

    {\centering
    $
        \begin{array}{ll}
            T_{max_1}[2] = & max(T_{max_1}[2], 2).
        \end{array}
    $
    
    }

%%%%%%%%%%%%%%%%%%%%%%%%%%%%%%%%%%%%%%%%%%%%%%%%%%%%%% PSEUDOCODE

\begin{algorithm}[!h]
\begin{footnotesize}
    \SetKwInOut{Input}{Input}
    \SetKwInOut{Output}{Output}
    \SetKwProg{Def}{def}{:}{}
    \SetKw{KwBy}{by}

    \SetAlgoLined
    \Input{Binary string $\omega$ with length $l$}
    \Output{$T_{max_1}$ and $T_{max_0}$ index tables}

    \BlankLine
    \Def{index($start, end, from, table$)}{
        \tcp{add elements to the index table for each substring}
        \If{$(end - start) \ \% \ 2 \ != 0$}{
            \For{$i\leftarrow from$ \KwTo $(end-start)$ \KwBy $2$}{
                $window = summed[start+i][0]-summed[start-1][0]$\;
                $count = summed[start+i][1]-summed[start-1][1]$\;
                \If{$\omega[start]==0$}{
                    $count -=  window$\;
                }
                $table[window] = max(count, table[window])$\;
            }
        }
    }

    \BlankLine
    \Def{windownize($table$)}{
        \tcp{fill out the table}
        \For{$i\leftarrow l$ \KwTo $1$}{
            $table[i] = max(table[i], table[i+1]-1)$\;
        }
        \For{$i\leftarrow 1$ \KwTo $l$}{
            $table[i] = max(table[i], table[i-1])$\;
        }
    }

    \Begin{
        % \emph{counted is a list of each character run size in $\omega$}\;
        $counted = []$\;
        \tcp{counted is a list of each character run size in $\omega$}
        $runs = 1$\;
        \For{$i\leftarrow 2$ \KwTo $l$}{
            \If{$\omega[i] == \omega[i-1]$}{
                $runs+=1$\;
            }
            \Else{
                $counted.add(runs)$\;
                $runs=1$\;
            }
        }
        $counted.add(runs)$\;
        % \BlankLine
        % \emph{summed is a list of how many $1$s or until each position in $\omega$}\;
        $summed = []$\;
        \tcp{summed is a list of how many $1$s in each prefix from $\omega$}
        $window = 0$\;
        $tot\_1 = 0$\;
        \For{$i, v \in counted$ \tcp*[h]{the pair i $\leftarrow$ key, v $\leftarrow$ value}}{
            $window += v$\;
            \If{$(i+\omega[0]) \ \% \ 2 \ != 0$}{
                $tot\_1 += v$\;
            }
            $summed.add([window, tot\_1])$\;
        }
        \BlankLine

        \emph{builds and traverse the suffix tree}\;
        $tree = Suffix\_Tree(counted)$; \tcp{Ukkonen's Algorithm $O(n)$}
        \For{$factor \in tree$ \tcp*[h]{Pre-order traversal}}{
            \If{$factor[0] + \omega[0] \ \% \ 2 == 0$ \tcp*[h]{substring starts in $0$}}{
                $index(factor.start, factor.end, factor.parent\_node.end, T_{max_0})$\;
            }
            \Else(// substring starts in $1$){
                $index(factor.start, factor.end, factor.parent\_node.end, T_{max_1})$\;
            }
        }
        $windonize(T_{max_1})$\;
        $windonize(T_{max_0})$\;
        \Return{$T_{max_1}, T_{max_0}$}
    }
\caption{\textsc{BJI by Suffix Tree}\label{alg:suffix_tree}}
\end{footnotesize}
\end{algorithm}

%%%%%%%%%%%%%%%%%%%%%%%%%%%%%%%%%%%%%%%

    % Let us illustrate with some steps from the suffix tree in \autoref{fig:partial-suffix-tree}. Select node $3$. Its substring is 'b', which starts with $1$. Indexing:
    % $
    %     \begin{array}{lcll}
    %         T_{max_1}[b] & \mapsto & T_{max_1}[2] = & max(T_{max_1}[2], 2).
    %     \end{array}
    % $

    We traverse the tree in pre-order, meaning that parents are already indexed when we reach their children. If the node substring length is even, then it starts and ends with the same digit. We only need to index substrings that enclose runs of the same digit, which occurs if the substring length is odd or if the difference between the parent's length and the node's length is greater than~$1$. In other words, if the difference between the parent's length and the node's length is $1$, then this node is adding only a run of $0$s or a run of $1$s. If the length of the node is even, then it adds a run of a different digit than the one we are currently indexing, and we can skip this node. Next, consider node~$7$. Its substring is '2-1', and the parent's substring is '2'. Since the substring length is even and the difference between the parent's length and the node's length is not greater than $1$, we do not need to index it. Moving on to node $1$. Its substring is '2-1-2-2-2-1-3-1-2-3-2-1-3', and the parent's substring is '2-1', so we index it starting at the parent.

{\centering
    $
        \begin{array}{rl}
             T_{max_1}[2+1+2] = & max(T_{max_1}[5], 2+2) \\
             T_{max_1}[2+1+2+2+2] = & max(T_{max_1}[9], 6) \\
             T_{max_1}[2+1+2+2+2+1+3] = & max(T_{max_1}[13], 9) \\
             T_{max_1}[2+1+2+2+2+1+3+1+2] = & max(T_{max_1}[16], 11) \\
             \cdots
        \end{array}
    $

}

    We have reached a leaf. The pre-order traversal will return to the previous parent and then move to the next child node, which is $17$. Its substring is '2-1-3' and its parent node is '2-1'. Notice that, although the difference between the parent's length and the node's length is not greater than $1$, the node's length is odd. Therefore, we index it. 
    $
        \begin{array}{ll}
            T_{max_1}[6] = & max(T_{max_1}[6], 5).
        \end{array}
    $

    We proceed to index each node of the tree until the pre-order traversal ends. If the node substring starts with the digit $0$, we index it to the table $T_{max_0}$ instead.

    When assigning values to the index table, it is important to highlight that we always check if the current indexed value is not already greater than the new value: $T[\alpha] = \max(T[\alpha], \beta)$.

    Traversing a suffix tree means obtaining each unique suffix of the string. We use it to index each prefix of each of those suffixes from the start to the end of runs of the same digit, thus achieving the same result as the Cunha et al.'s algorithm, but avoiding duplicated substrings.

    \paragraph{Time complexity analysis.} \label{time}
    
    Recall that the worst-case scenario for Cunha's algorithm is a pattern of interspersed $1$s and $0$s: $101010\cdots$. But for our algorithm, this becomes the best case, as it is full of repetition. Notice that there exists only one substring for each $l$-sized window starting and ending in a run of $1$s ($1: 1; \ 3: 101; \ 5: 10101 \cdots$), thus achieving linear time using the suffix tree.

    One might initially assume that the worst-case scenario for our algorithm is when there are no repetitions. Let $\omega$ be a binary word of length $n$ over the alphabet $\{0, 1\}$ with no repetitions, meaning it has distinct run lengths for each digit. Since the complexity analysis is based on the number of runs ($\rho$) rather than the specific results of the index table, the order of these runs does not affect the complexity. For clarity, consider two examples: $101100111000$ and $111100011000$. While the former does not contain a substring of size $5$ with $4$ $1$s, the latter does. However, this difference is irrelevant to the complexity, as it depends solely on $\rho$. For simplicity, the runs can be sorted, resulting in a word such as $\omega = 10110011100011110000\cdots$. The sorted version maximizes the number of runs, making it the worst-case scenario for the algorithm when there are no repetitions.

    Now we can establish $n$ in terms of $\rho$:

    {\centering
    $
        \begin{array}{ccc}
        n =  \sum_{i=1}^{\rho}2i  =  2\frac{\rho(\rho+1)}{2} = \rho^2+\rho, & \mapsto & \rho  \simeq  \sqrt{n}.  \\
        \end{array}
    $

    }

    Since the overall time complexity for our algorithm is $O(n+\rho^2)$ and $\rho \simeq \sqrt{n}$, then when no repetitions occur, the algorithm achieves $O(n+n) = O(n)$ time.

    The worst case for our algorithm shifts to the worst space complexity of a suffix tree, as it will produce the maximum number of different factors depending on $n$. The worst space complexity for a suffix tree is already known to be the Fibonacci word, or the analogous rabbit sequence. This is a sequence of strings obtained by considering $s_0 = 0$, $s_1 = 01$, $s_n = s_{n-1}\cdot s_{n-2}$, where $\cdot$ denotes concatenation of two strings.

    It is counterintuitive why the Fibonacci word or the rabbit sequence represents the worst-case scenario for space complexity in the suffix tree and time complexity for our algorithm. To illustrate this, consider the following example: compare two binary strings with the same size: $S_1 = 10110101$ and $S_2 = 10110111$. $S_2$ has no repetitions of factors enclosing runs of $1$s and $S_1$ has two repetitions; however, $S_2$ has $\rho = 3$ while $S_1$ has $\rho = 4$. Calculating unique factors enclosing runs of $1$s, we have:

    {\centering
    $
        \begin{array}{ccccccccc}
             S_1 \mapsto & \frac{4(4+1)}{2} & = 10 - 2 & = 8, & &
             S_2 \mapsto & \frac{3(3+1)}{2} & = 6 - 0 & = 6.
        \end{array}
    $

    }

    Even though $S_2$ has fewer repetitions, it comes at the cost of word space that could otherwise be used to increase the number of runs. Therefore, if we aim to maximize the number of unique factors, it is achieved through an equilibrium between repetitions and run size, which results in the Fibonacci word or the rabbit sequence.

    One may notice that the Fibonacci word has many repeated runs of $1$, since each run has length $1$. This can be easily explained by the formula:

    {\centering
    $
        \alpha^2+\beta^2 < (\alpha+\beta)^2 \ | \ \forall \alpha, \beta > 0.
    $

    }

    Let $\alpha$ and $\beta$ be the variance of runs of each digit. The uniqueness of substrings relies on the variance of the runs. This formula shows that varying runs of two digits is less effective than allowing only one of the digits to vary to create unique factors.
    
    Although the Fibonacci word and rabbit sequence are the worst-case scenarios for our algorithm, they still exhibit many repetitions, even more than the average binary string. Therefore, we demonstrate in a comparison table (\autoref{tab:comparacao_fibo}) that they are favorable for our algorithm, as the other one is bounded by~$\rho^2$.

    In the context of average-case analysis for the proposed algorithm, it is crucial to understand the behavior of runs in a binary string. The average number of runs in a binary string directly influences the complexity of the algorithms. Based on this, we present the following theorem, which provides the average number of runs in a binary string of length $n$. This result will serve as the foundation for the average-case complexity analysis of the algorithms.

    % \azul{
        \begin{theorem}\label{thm:averageruns}
            The average number of runs $\rho$ in a string with size $n$ is $\frac{n}{4}$.
        \end{theorem}
        \begin{proof}
            We begin by noting that each run of $1$s must contain at least one $1$, and each run must be separated by at least one $0$. Thus, we reserve $\rho$ digits for the $1$s (one for each run) and $\rho-1$ digits for the $0$s, which are placed between the runs. 

            This leaves us with $\kappa = n - \rho - (\rho-1) = n - 2\rho + 1$ elements remaining to distribute. These elements can be assigned either to the $\rho$ groups of $1$s or to the $\rho-1 + 2 = \rho+1$ groups of $0$s, as zeros can also be placed at the edges of the string.
            
            The number of valid distributions is therefore given by the following formula:
            \[
            \sum_{i=0}^{\kappa} \binom{i + \rho - 1}{\rho - 1} \cdot \binom{\kappa - i + \rho + 1 - 1}{\rho + 1 - 1}.
            \]
            
            Here, the first binomial coefficient counts the ways to distribute $i$ extra elements among the $\rho$ groups of $1$s, while the second binomial coefficient counts the ways to distribute the remaining $\kappa - i$ elements among the $\rho + 1$ groups of $0$s.
            
            The maximum number of runs for any binary string of length $n$ is $\frac{n}{2}$, since each run requires at least one $1$ and one $0$, and there are $n$ total elements. 
            
            Since the formula reflects a binomial distribution between $1$ and $\frac{n}{2}$, and the mean of a binomial distribution occurs at its midpoint, the expected number of runs $\rho$ for a string of length $n$ is:
            \[
            \frac{1}{2} \cdot \frac{n}{2} = \frac{n}{4}.
            \]~\qed
        \end{proof}
    
        We prove in \autoref{thm:averageruns} that the average number of runs for a given binary string is $\frac{n}{4}$, therefore, as a corollary, we can also prove the average-case time complexity for the JBM2017 algorithm:

        \[
            \rho^2 = \left( \frac{n}{4} \right)^2 = \frac{n^2}{16} = \Theta(n^2)
        \]

        Our proposed algorithm improves efficiency by avoiding redundant counting of repeated substrings. Consequently, as established in \autoref{thm:averageruns}, the average-case time complexity can be expressed as $\Theta(n^2) - \Theta(r)$, where $r$ denotes the average number of repeated substrings. To determine $r$, we utilize an alternative data structure: the suffix trie.

        The suffix trie retains the same fundamental structure as a typical suffix tree, but with an additional node inserted between each digit, making every edge unary. In the proposed algorithm, each digit in the tree is processed individually. Consequently, the average time complexity of the algorithm aligns more closely with the average space complexity of the suffix trie, which is $O(n^2)$. This indicates that the number of repetitions grows relatively slowly with respect to the string length $n$.
    % }

    % {\color{red}
    %     Moreover, the first two elements of the Fibonacci word are $0$ and $01$, and since every concatenation will start in $0$, then at every new step the number of runs of $1$ will follow a fibonacci sequence. Let $\rho$ be the number of runs of $1$, and $F_i$ the $i$th number of the fibonacci sequence:
    %     \[
    %         n=F_i \implies \rho=F_{i-1}
    %     \]
    %     \[
    %         \rho = \frac{n}{1.6}
    %     \]
    % }

    \paragraph{Proving $\rho^2$ as lower bound for indexing table.} \label{p2}

    We know that:
        $T[i] = T[i-1]$ or $T[i] = T[i-1]+1$.
        If $T[i] = T[i-1]$, then there is a max($i$-factor) that starts or ends in $1$. It could exist or not exist an $i$-factor starting and ending in $1$.
        If $T[i] = T[i-1]+1$, then any max($i$-factor) must start with $1$ and end with $1$. Otherwise, we could remove any edge $0$ and index $T[i-1] = T[i]$, which would lead to a contradiction.

    %We only need to look strings starting and ending in 1%
    \begin{lemma}\label{lm:runs}
    To build an index table based on comparing runs of $1$s, the optimal strategy is to index only factors of the word that begin and end in $1$ with no subsequent $1$s.
        % It is faster to only index factors of the word that start and end in $1$ with no subsequent $1$s.
    \end{lemma}
    \begin{proof}
        Let $\omega$ be a binary word of length $n$. For each $i$, $0 < i \leq n$, there are $n-i+1$ $i$-length factors of $\omega$. We need to build an index table of size $n$, where for each $i$, $0 < i \leq n$, we only need to index an $i$-factor with the maximum number of $1$s.
        If $T[i] = T[i-1]+1$, then the max($i$-factor) starts and ends in $1$, as established, and we would index it. But if $T[i] = T[i-1]$, with no $i$-factors starting and ending in $1$, at the end of the algorithm $T[i]$ would be empty, and we could set $T[i] = T[i-1]$, ignoring all $i$-factors. This will save us $n-i+1-1 = n-i$ operations. Since $i \leq n$, it will never be more expensive to use neighbor indexed values.
        Now, if $T[i] = T[i+1]-1$, analogously it is faster to set $T[i] = T[i+1]-1$ than to search for max($i$-factors).~\qed
    \end{proof}

    We have shown in \autoref{lm:runs} that any algorithm that constructs an index table is faster when it only utilizes factors starting and ending in $1$.

    \begin{theorem}\label{thm:lowerlimit}
    $\Omega(n+\rho^2)$ is a lower bound for building an index table based on comparing runs of $1$s.
        % $\Omega(n+\rho^2)$ is a lower bound for building an index table based on runs of $1$s.
    \end{theorem}
    \begin{proof}
        Let $AxByCzD\cdots$ be a binary encoded string, where $A,B,C\cdots$ are the sizes of each run of $1$s and $x,y,z\cdots$ are the sizes of each run of $0$s. Let $\rho$ be the number of runs of $1$ in the word. \autoref{lm:runs} shows that we only need to index $max(A, B, C, D \cdots)$, but since we do not have an ordered list, then we have to spend $\Omega(\rho)$ to find it. Now, depending on the values of each run, we could have $|A|+|x|+|B| = |B|+|y|+|C|+|z|+|D|$, but suppose that each $AxB, ByC, CzD \cdots$ all have the same size $s$ in the decoded binary string. Then we need to index $T[s] = \max_1(AxB, ByC, CzD \cdots)$. Again we do not have an ordered list, so to find $\max_1(AxB, ByC, CzD \cdots)$ we will need at least $\Omega(\rho-1)$. This argument is analogous for runs of $1$s three by three, four by four, and so on. Therefore, we need to index at least from the start of each run of $1$ to the end of each one. The time complexity for this is $\Omega(n+\sum_{i=1}^{\rho}i) = \Omega(n+\rho^2)$.~\qed
    \end{proof}

    \section{Practical results and discussions} \label{results}

        Now, we present practical results by comparing our suffix tree algorithm\footnote{Implementations available at: \\ \url{https://github.com/mariozenmedina/jumbled-pattern-matching/blob/master/sft_vs_p2.py}} 
        with Cunha et al.'s algorithm. The process begins by selecting a size for the string to be indexed. Subsequently, several random binary strings of this size are generated, and then each algorithm is applied. The time taken to construct $T_{max_1}$ and $T_{max_0}$ tables, as described in \autoref{tab:comparacao}, is then displayed.

        \setlength{\arrayrulewidth}{0.3mm}
        \setlength{\tabcolsep}{2.5pt}

%%%AQUI%%%%%
% %\vspace{-.5cm}
    \begin{table}[!h]        
    
\begin{multicols}{2}
        \begin{center}
            \begin{tabular}{ | c | c | c | c | }
                \hline
                    \multicolumn{4}{|c|}{$1,000$ strings with length $1,000$} \\
                \hline
                     Algorithm & Min & Max & Avg \\ \hline
                     JBM2017 & 0.0470s & 0.1840s & 0.0642s \\
                     SfTree & 0.0186s & 0.0905s & 0.0244s  \\
                \hline
            \end{tabular}
        \end{center}

% \columnbreak

        \begin{center}
            \begin{tabular}{ | c | c | c | c | }
                \hline
                    \multicolumn{4}{|c|}{$1,000$ strings with length $5,000$} \\
                \hline
                     Algorithm & Min & Max & Avg \\ \hline
                     JBM2017 & 1.3350s & 5.6991s & 1.6792s  \\
                     SfTree & 0.2023s & 0.9228s & 0.2750s \\
                \hline
            \end{tabular}
        \end{center}

\columnbreak

        \begin{center}
            \begin{tabular}{ | c | c | c | c |  }
                \hline
                    \multicolumn{4}{|c|}{$1,000$ strings with length $10,000$} \\
                \hline
                     Algorithm & Min & Max & Avg \\ \hline
                     JBM2017 & 5.1280s & 10.3999s & 5.6467s  \\
                     SfTree & 0.6754s & 1.1174s & 0.7536s  \\
                \hline
            \end{tabular}
        \end{center}
    %\vspace{-1cm}

\end{multicols}

\caption{A time comparison for indexing 1,000 random binary strings with different lengths is presented between Cunha et al.'s algorithm~\cite{cunha2017faster} (JBM2017) and our proposed suffix tree-based algorithm (SfTree), including the minimum, maximum, and average processing times. \label{tab:comparacao}}

%\vspace{-1cm}
    \end{table}

\begin{table}[!h]       
\begin{multicols}{2}
            \begin{center}
                \begin{tabular}{ | c | c | c | }
                    \hline
                        \multicolumn{3}{|c|}{Interspersed string $(010101\cdots)$} \\
                    \hline
                         Algorithm & Length: 10,000 & Length: 100,000 \\ \hline
                         JBM2017 & 23.1981s & 2213.4648s \\
                         SfTree & 0.0415s & 0.3587s  \\
                    \hline
                \end{tabular}
            \end{center}
            \caption{Time comparison for indexing between Cunha et al.'s algorithm~\cite{cunha2017faster} (JBM2017) and our proposed algorithm by using suffix trees (SfTree).\label{tab:comparacao_best}}
%\vspace{-.5cm}

\columnbreak
            \begin{center}
                \begin{tabular}{ | c | c | c | }
                    \hline
                        \multicolumn{3}{|c|}{Fibonacci word $(0100101001001 \cdots)$} \\
                    \hline
                         Algorithm & Length: 5,000 & Length: 50,000 \\ \hline
                         JBM2017 & 3.3024s & 378.1867s \\
                         SfTree & 0.2491s & 17.6050s  \\
                    \hline
                \end{tabular}
            \end{center}
            \caption{Time comparison for Fibonacci binary word between Cunha et al.'s algorithm~\cite{cunha2017faster} (JBM2017) and our proposed algorithm by using suffix trees (SfTree).\label{tab:comparacao_fibo}}
    %\vspace{-1cm}
\end{multicols}

        \end{table}

        \begin{figure}[!h]
            \centering
            \includegraphics[width=10.5cm]{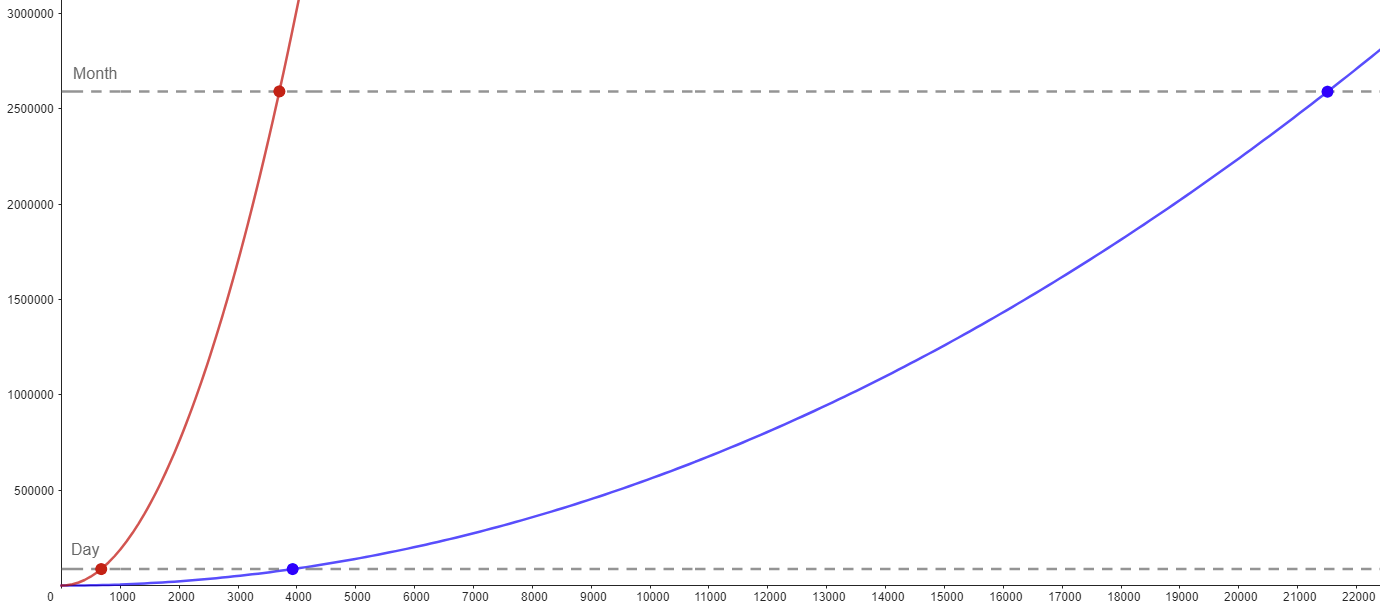}
            \caption{ 
                Execution time comparison between JBM2017 and SFTree algorithms. The y-axis represents execution time in seconds, while the x-axis represents the input size in thousands of digits. The graph illustrates the simulated asymptotic growth curves for both algorithms.
                \label{fig:asymptotic-simulation}
            }
        \end{figure}

        In \autoref{fig:asymptotic-simulation} we show that both algorithms exhibit quadratic growth, though with differing slopes of increase. The use of a suffix tree, provides a significant advantage in terms of practical performance. By leveraging the tree structure, the algorithm minimizes redundant operations and efficiently organizes substrings for indexing. This approach facilitates rapid access to substrings and their positions.

        Comparative tables and figures (\autoref{tab:comparacao}, \autoref{tab:comparacao_best}, \autoref{tab:comparacao_fibo}, \autoref{fig:asymptotic-simulation} and \autoref{fig:asymptotic-comparison}) highlight a substantial performance gap, with the suffix tree-based algorithm being markedly faster.
        
        While both algorithms have a quadratic time complexity in the average case, the vectorization applied in the suffix tree implementation offers a significant advantage. Vectorization refers to the process of optimizing the algorithm to perform operations on multiple data elements simultaneously, rather than iterating over them one by one. This allows for more efficient use of CPU resources, especially when handling large datasets.
        
        In the case of the suffix tree implementation, vectorization optimizes memory access to occur in linear time, where $n$ represents the number of nodes in the tree, and $2n$ is the maximum number of nodes. During traversal of the suffix tree, for each new node, its corresponding edge substring is indexed all at once rather than iterating over each individual digit.

        By processing multiple elements simultaneously, vectorization minimizes the impact of memory latency, a common bottleneck in algorithm performance. As a result, the suffix tree-based algorithm outperforms its counterparts, making it more efficient in practical applications.

        \begin{figure}[!h]
            \centering
            \includegraphics[width=9.5cm]{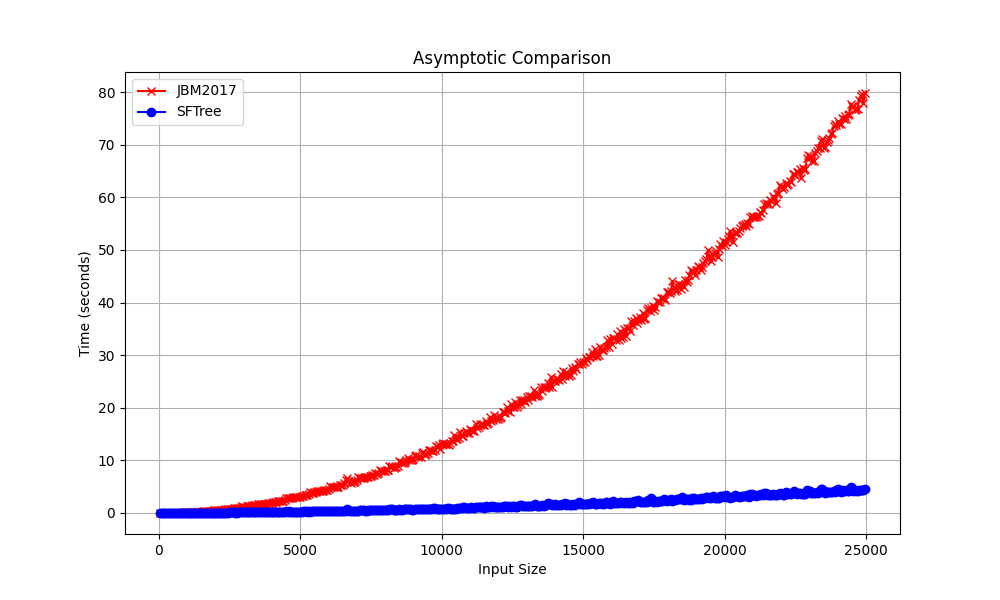}
            \caption{ 
                Asymptotic comparison between JBM2017 and SFTree in range(50, 25000, 50)
                \label{fig:asymptotic-comparison}
            }
        \end{figure}
        
        We highlight the advantages of our suffix tree algorithm:

        \begin{itemize}
            \item The suffix tree is a well-known data structure and can be used for various other applications within the same string, such as search, data compression, exact string matching, and others.
            \item It has the capability to construct a generalized tree for multiple strings, which saves time by building multiple index tables and reusing repetitions between those strings. %This feature is also beneficial for solving the Longest Common Subsequence (LCS) problem between strings.
            \item The algorithm is compatible with any traversal order and can be integrated into other traversal applications with minimal additional cost.
            % \item Despite not being fully optimized, it exhibits the same time complexity and practically the same execution time as Cunha et al.'s algorithm in the worst-case.
            \item There are no instances where it is slower than the other algorithm, but it can be much faster when the string contains sufficient repetitions, fluctuating between quadratic and linear time.
        \end{itemize}

        The execution time for Cunha et al.'s algorithm is particularly relevant in its worst-case scenario, where it remains quadratic, while the suffix tree allows us to achieve linear time, as shown in \autoref{tab:comparacao_best}. Another significant scenario is the Fibonacci word, which represents our worst-case scenario. As explained in the time complexity analysis in \autoref{algorithm}, despite being counter-intuitive, the time difference is more pronounced in our worst-case scenario due to repetitions, highlighting our advantages.

        The cost of processing ($p$) refers to the time required for CPU operations, while the cost of memory access ($m$) represents the latency and transfer time for retrieving data from memory. Memory access is significantly more expensive than processing, often by a factor of 10 to 100, depending on hardware architecture and caching mechanisms. Comparing the complexities of the algorithms, JBM2017 operates in $O(n^2 \cdot p + n^2 \cdot m)$, SFTree in $O(n^2 \cdot p + n \cdot m)$, and the 3SUM~\cite{chan2015clustered} in $O(n^{1.864} \cdot p + n^{1.864} \cdot m)$. While 3SUM~\cite{chan2015clustered} has a lower asymptotic growth due to its $O(n^{1.864})$ complexity, the SFTree algorithm takes advantage of its linear memory access term $O(n \cdot m)$, which can lead to better performance for most practical inputs. This advantage becomes especially pronounced in cases where the cost of memory access ($m$) dominates the cost of processing ($p$), as is typical in real-world systems.

        Considering the relationship between processing cost ($p$) and memory access cost ($m$), we observe significant differences in the crossover point where 3SUM becomes more efficient than SFTree. If $m = 10p$, 3SUM outperforms SFTree for input sizes larger than $4.5 \times 10^7$. However, if $m = 100p$, the crossover occurs only for input sizes exceeding $5.5 \times 10^{14}$. These results highlight the impact of memory access cost on the practical performance of the algorithms.

\bibliographystyle{abbrv}
\bibliography{cocoon}

\newpage

\section*{Declarations}

\paragraph{Funding}

Financial support from FAPERJ E-26/201.372/2022, and CNPq 406173/2021-4.

\paragraph{Conflict of interest}
Not applicable.

\paragraph{Ethics Approval}
Not applicable.

\paragraph{Consent to participate}
Not applicable.

\paragraph{Consent for publication}
Not applicable.

\paragraph{Data availability}

Only data generated by the presented algorithms were used. They can be similarly replicated by new executions of the code, but no particular dataset is provided.

\paragraph{Code availability}

All the algorithms are explained and the Python code used to generate simulations and comparisons is available in: \url{https://github.com/mariozenmedina/jumbled-pattern-matching/blob/master/suffixtree_vs_p2.py}

\paragraph{Author contributions}
The authors contributed equally to this work.

% \appendix

% \section{Codes comparison}
% \label{appendix:comparision}

% \autoref{alg:indexComparison} compares the equivalence between Cunha et al.'s algorithm and our proposed one using suffix trees, with both algorithm's tables for correctness (developed and presented in \emph{Python 3}). 

% \lstinputlisting[caption = Comparison between Cunha et al.'s algorithm (JBM 2017)~\cite{cunha2017faster} and the algorithm we developed in this work (SF Tree).,label=alg:indexComparison]{suffixtree_vs_p2.py}

%\begin{acknowledgements}
%\end{acknowledgements}



% \begin{thebibliography}{}
% %
% % and use \bibitem to create references. Consult the Instructions
% % for authors for reference list style.
% % Format for Journal Reference
% Author, Article title, Journal, Volume, page numbers (year)
% Author, Book title, page numbers. Publisher, place (year)
% \end{thebibliography}

\end{document}